\documentclass[11pt,reqno]{amsart}
\usepackage[top=1.2in, bottom=1.2in, left=1.2in, right=1.2in]{geometry}
\usepackage{amsfonts, amssymb, amsmath, mathtools, dsfont}
\usepackage[linktoc=page, colorlinks, linkcolor=blue, citecolor=blue]{hyperref}
\usepackage{enumerate, commath}

\usepackage{graphicx,xcolor, abstract}
\usepackage[font=small]{caption} 

\usepackage{algorithm,algpseudocode}

\numberwithin{equation}{section}

\DeclareMathOperator{\E}{\mathbb{E}}

\DeclareMathOperator{\Var}{Var}

\DeclareMathOperator{\diam}{diam}

\DeclareMathOperator{\sign}{sign}
\DeclareMathOperator{\Lap}{Lap}
\DeclareMathOperator{\flux}{flux}

\renewcommand{\Pr}[2][]{\mathbb{P}_{#1} \left\{ #2 \rule{0mm}{3mm}\right\}}

\def \R {\mathbb{R}}

\def \Z {\mathbb{Z}}

\def \MM {\mathcal{M}}

\def \XX {\mathcal{X}}
\def \YY {\mathcal{Y}}
\def \ZZ {\mathcal{Z}}

\def \e {\varepsilon}
\def \d {\delta}
\def \l {\lambda}
\def \s {\sigma}

\newtheorem{theorem}{Theorem}[section]
\newtheorem{proposition}[theorem]{Proposition}
\newtheorem{corollary}[theorem]{Corollary}
\newtheorem{lemma}[theorem]{Lemma}

\newtheorem{definition}[theorem]{Definition}

\theoremstyle{remark}
\newtheorem{remark}[theorem]{Remark}

\newcommand{\cF}{\mathcal{F}}
\newcommand{\cX}{\mathcal{X}}
\newcommand{\cY}{\mathcal{Y}}
\newcommand{\cN}{\mathcal{N}}
\newcommand{\dd}{\mathrm{d}}
\newcommand{\dbl}{d_{\mathrm{BL}}}
\newcommand{\Lip}{\mathrm{Lip}}

\renewcommand{\paragraph}[1]{\subsection*{#1}}

\title{Algorithmically Effective Differentially Private Synthetic Data}

\author{Yiyun He}
\address{Department of Mathematics, University of California Irvine, Irvine, CA 92697}
\email{yiyunh4@uci.edu}

\author{Roman Vershynin}
\address{Department of Mathematics, University of California Irvine, Irvine, CA 92697}
\email{rvershyn@uci.edu}

\author{Yizhe Zhu}
\address{Department of Mathematics, University of California Irvine, Irvine, CA 92697}
\email{yizhe.zhu@uci.edu}

\usepackage{times}

\begin{document}

\maketitle

\begin{abstract}%
 We present a highly effective algorithmic approach for generating $\e$-differentially private synthetic data in a bounded metric space with near-optimal utility guarantees under the  1-Wasserstein distance. In particular, for a dataset $\mathcal X$ in the hypercube $[0,1]^d$, our algorithm  generates synthetic dataset $\mathcal Y$ such that the expected 1-Wasserstein distance between the empirical measure of $\mathcal X$ and $\mathcal Y$ is $O((\e n)^{-1/d})$ for $d\geq 2$, and is $O(\log^2(\e n)(\e n)^{-1})$ for $d=1$. The accuracy guarantee is optimal up to a constant factor for $d\geq 2$, and up to a logarithmic factor for $d=1$. Our algorithm has a fast running time of $O(\e d n)$ for all $d\geq 1$ and demonstrates improved accuracy compared to the method  in \cite{boedihardjo2022private} for $d\geq 2$.
\end{abstract}


\section{Introduction}

Differential privacy has become the benchmark for privacy protection in scenarios where vast amounts of data need to be analyzed. The aim of differential privacy is to prevent the disclosure of information about individual participants in the dataset. In simple terms, an algorithm that has a randomized output and produces similar results when given two adjacent datasets is considered to be differentially private. This method of privacy protection is increasingly being adopted and implemented in various fields, including the 2020 US Census  \cite{abowd2019census,hawes2020implementing,hauer2021differential} and numerous machine learning tasks \cite{dwork2014algorithmic}.

A wide range of data computations can be performed in a differentially private manner, including regression \cite{chaudhuri2008privacy}, clustering \cite{su2016differentially}, parameter estimation \cite{duchi2018minimax}, stochastic gradient descent \cite{song2013stochastic}, and deep learning \cite{abadi2016deep}. However, many existing works on differential privacy focus on designing algorithms for specific tasks and are restricted to queries that are predefined before use. This requires expert knowledge and often involves modifying existing algorithms.

One promising solution to this challenge is to generate a synthetic dataset similar to the original dataset with guaranteed differential privacy \cite{hardt2012simple,blum2013learning,jordon2019pate,bellovin2019privacy,boedihardjo2022covariance,boedihardjo2022privacy,boedihardjo2022private}. As any downstream tasks are based on the synthetic dataset, they can be performed without incurring additional privacy costs. 

\subsection{Private synthetic data} 

Mathematically, the problem of generating private synthetic data can be defined as follows. Let $(\Omega,\rho)$ be a metric space. Consider a dataset $\mathcal X=(X_1,\dots, X_n)\in \Omega^n$. Our goal is to construct an efficient randomized algorithm that outputs differentially private synthetic data $\YY=(Y_1,\dots,Y_m)\in \Omega^m$  such that the two empirical measures
\[  \mu_{\XX}=\frac{1}{n} \sum_{i=1}^n \delta_{X_i} \quad \text{and} \quad \mu_{\YY}=\frac{1}{m} \sum_{i=1}^m \delta_{Y_i}\]
 are close to each other. We measure the utility of the output by $\E W_1(\mu_{\XX},\mu_{\YY})$, where $W_1(\mu_{\XX},\mu_{\YY})$ is  the 1-Wasserstein distance, and the expectation is taken over the randomness of the algorithm.  The Kantorovich-Rubinstein duality  (see, e.g., \cite{villani2009optimal}) gives an equivalent representation of the 1-Wasserstein distance between two measures $\nu_{\XX}$ and $\mu_{\YY}$:
\begin{equation}
\label{eq: KR duality}
    W_1(\mu_{\XX},\mu_{\YY}) = \sup_{{\mathrm{Lip}(f)}\leq 1} \left(\int f\dd\mu_{\XX}-\int f\dd \mu_{\YY}\right),
\end{equation}
where the supremum is taken over the set of all  $1$-Lipschitz functions on $\Omega$. Since many machine learning algorithms are Lipschitz \cite{von2004distance,kovalev2022lipschitz,bubeck2021universal,meunier2022dynamical}, Equation~\eqref{eq: KR duality} provides a uniform accuracy guarantee for a wide range of machine learning tasks  performed on synthetic datasets whose empirical measure is close to $\mu_{\XX}$ in the 1-Wasserstein distance.

\subsection{Main results}
The most straightforward way to construct differentially private synthetic data  is to add independent  noise to the location of each data point.
However, this method can result in a significant loss of data utility as the amount of noise needed for privacy protection may be too large \cite{domingo2021limits}.  Another direct approach could be to add noise to the density function of the empirical measure of $\XX$, by dividing $\Omega$ into small subregions and perturbing the true counts in each subregion. However, Laplacian noise may perturb the count in a certain subregion to negative, causing the output to become a signed measure. To address this issue, we introduce an algorithmically effective method called the \textit{Private Measure Mechanism}.

\paragraph{Private Measure Mechanism (PMM)}
PMM makes the count zero if the noisy count in a subregion is negative.  Instead of a single partition of $\Omega$, we consider a collection of binary hierarchical partitions on $\Omega$ and add inhomogeneous noise to each level of the partition. However, the counts of two subregions do not always add up to
the count of the region at a higher level. We develop an algorithm that enforces the consistency of counts in regions at different levels.  
PMM has $O(\e dn)$ running time while the running time of the approach in  \cite{boedihardjo2022private} is  polynomial in $n$.

The accuracy analysis of PMM uses the hierarchical partitions to estimate the 1-Wasserstein distance in terms of the multi-scale geometry of $\Omega$ and the noise magnitude  in each level of the partition. In particular, when $\Omega=[0,1]^d$, by optimizing the choice of the hierarchical partitions and noise magnitude, PMM achieves  better accuracy compared to \cite{boedihardjo2022private} for $d\geq 2$. The accuracy is optimal  rate up to a constant factor for $d\geq 2$, and up to a logarithmic factor for $d=1$. We state it in the next theorem. 

The  hierarchical partitions appeared in many previous works on the approximation of distributions under Wasserstein distances  in a non-private setting, including \cite{ba2011sublinear,dereich2013constructive,weed2019sharp}. In the differential privacy literature, the hierarchical partitions are also closely related to the binary tree mechanism  \cite{dwork2010differential,chan2011private} for differential privacy under continual
observation. However, the accuracy analysis of the two mechanisms is significantly different. In addition, the  TopDown algorithm in the 2020 census \cite{abowd20222020} also has a similar hierarchical structure and enforces consistency, but the accuracy analysis of the algorithm is not provided in \cite{abowd20222020}.

\begin{theorem}[PMM for data in a hypercube]\label{thm: PMM}
    Let $\Omega=[0,1]^d$ equipped with the $\ell^{\infty}$ metric. PMM outputs an  $\e$-differentially private synthetic dataset $\YY$ in time $O(\e dn)$  such that
    \[\E W_1(\mu_\cX,\mu_{\YY}) \leq \left\{\begin{aligned}
    &C\log^2(\e n)(\e n)^{-1}&\quad \textrm{ if } d=1,\\
    &C {(\e n)^{-\frac{1}{d}}}&\quad \textrm{ if } d\geq 2.
\end{aligned}\right.\]
\end{theorem}

\paragraph{Private Signed Measure Mechanism (PSMM)}

In addition to PMM, we introduce an alternative method, the \textit{Private Signed Measure Mechanism},  that achieves optimal accuracy rate on $[0,1]^d$ when $d\geq 3$ in  $\text{poly}(n)$ time. The analysis of PSMM  is not restricted to 1-Wasserstein distance, and it can be generalized to provide a uniform utility guarantee of other function classes. 


We first partition the domain $\Omega$ into $m$ subregions $\Omega_1,\dots, \Omega_m$. Perturbing the counts in each subregion with i.i.d. integer Laplacian noise gives an unbiased approximation of $\mu_{\YY}$ with a  signed measure $\nu$. Then we find the closest probability measure $\hat{\nu}$ under the bounded Lipschitz distance by solving a linear programming problem.


In the proof of accuracy for PSMM, one ingredient is to estimate the Laplacian complexity of the Lipschitz function class on $\Omega$ and connect it to the 1-Wasserstein distance. This type of argument is similar in spirit to the optimal matching problem for two sets of random points in a metric space \cite{talagrand1992matching,talagrand2005generic,bobkov2021simple}. 
When $\Omega=[0,1]^d$, PSMM achieves the optimal accuracy rate $O((\e n)^{-1/d})$ for $d\geq 3$. For $d=2$,  PSMM achieves a near-optimal accuracy  $O(\log(\e n)(\e n)^{-1/2})$. For $d=1$, the accuracy  becomes $O((\e n)^{-1/2})$. 

For the case when $d=2$, we believe that the bound in Corollary \ref{cor: accuracy_direct_hypercube} could be improved to $C\sqrt{\log (\e n)}/\sqrt{\e n}$ by replacing Dudley's chaining  bound  in Proposition \ref{prop:Ln(f)} with the generic chaining bound in  \cite{talagrand2005generic,dirksen2015tail} involving the $\gamma_1$ and $\gamma_2$  functionals on $\Omega$. We will not pursue this direction in this paper.

\paragraph{Comparison to previous results}
\cite{ullman2011pcps} proved that it is NP-hard
to generate private synthetic data on the Boolean cube which approximately preserves all two-dimensional marginals, assuming the existence of one-way functions. There exists a substantial body of work for  differentially private synthetic data   with guarantees limited to accuracy bounds for a finite set of  specified queries \cite{barak2007privacy,thaler2012faster,dwork2015efficient,vadhan2017complexity,liu2021iterative,vietriprivate,boedihardjo2022private,boedihardjo2022private2,boedihardjo2023covariance}. 

\cite{wang2016differentially} considered differentially private synthetic data in $[0,1]^d$ with guarantees for any smooth queries with bounded partial derivatives of order $K$, and achieved an accuracy of $O(\e^{-1}n^{-\frac{K}{2d+K}})$. Recently, \cite{boedihardjo2022private} introduced a method based on superregular random walks to generate differentially private synthetic data with near-optimal guarantees in 
general compact metric spaces. In particular, when the dataset is in $[0,1]^d$, they obtain $\E W_1(\mu_\cX,\mu_{\YY})\leq C\log^{\frac{3}{2d}}(\e n) (\e n)^{-\frac{1}{d}}$. A corresponding lower bound of order $n^{-1/d}$ was also proved in \cite[Corollary 9.3]{boedihardjo2022private}.  PMM matches the lower bound up to a constant factor for $d\geq 2$, and up to a logarithmic factor for $d=1$.

In terms of computational efficiency, PMM runs in time $O(\e dn)$. This is more efficient compared to the  algorithm in \cite{boedihardjo2022private}.
 
\paragraph{Organization of the paper}
The rest of the paper is organized as follows. In Section~\ref{sec: preliminaries}, we introduce some background on differential privacy and distances between measures. We will first introduce and analyze the easier and more direct method PSMM before our main result. In Section~\ref{sec: psmm}, we describe PSMM in detail and prove its privacy and accuracy for  data in a bounded metric space, and detailed results are provided for the case for the hypercube. In Section~\ref{sec: pmm}, we introduce PMM  and analyze its privacy and accuracy. Optimizing the choices of noise parameters, we obtain the optimal accuracy on the hypercube with $O(\e dn)$ running time, which proves  Theorem~\ref{thm: PMM}. 

Additional proofs are included in Appendix~\ref{sec: appendix proof}. We use a variant of Laplacian distribution, called \textit{discrete Laplacian distribution}, in PMM and PSMM. The definition and properties of discrete Laplacian distribution are included in Appendix \ref{sec: appendix discrete}.

\section{Preliminaries}
\label{sec: preliminaries}
\paragraph{Differential Privacy}
We use the following definition from  \cite{dwork2014algorithmic}. A randomized algorithm $\MM$ provides \textit{$\e$-differential privacy} if for any input data $D,D'$ that differs on only one element (or $D$ and $D'$ are adjacent data sets) and for any measurable set $S\subseteq\mathrm{range}(\MM)$, there is 
    \[\frac{\Pr{\MM(D)\in S}}{\Pr{\MM(D')\in S}}\leq e^\e.\]
    Here the probability is taken from the probability space of the randomness of $\MM$.

\paragraph{Wasserstein distance}
Consider a metric space $(\Omega,\rho)$ with two probability measures $\mu,\nu$. Then the 1-Wasserstein distance (see e.g., \cite{villani2009optimal} for more details) between them is defined as 
\[W_1(\mu,\nu):=\inf_{\gamma\in\Gamma(\mu,\nu)} \int_{\Omega\times\Omega}\rho(x,y)\dd\gamma(x,y),\]
where $\gamma(\mu,\nu)$ is the set of all couplings  of $\mu$ and $\nu$.

\paragraph{Bounded Lipschitz distance}
Let $(\Omega,\rho)$ be a bounded metric space. The  \textit{Lipschitz norm} of a function $f$ is defined as
\[\norm{f}_{\Lip} := \max \left\{\Lip(f),\;\frac{\norm{f}_\infty}{\diam(\Omega)}\right\},\]  where $\Lip(f)$ is the Lipschitz constant of $f$. Let $\cF$ be the set of all Lipschitz functions $f$ on $\Omega$ with $\norm{f}_{\Lip}\leq 1$. 
For signed measures $\mu,\nu$, we define the \emph{bounded Lipschitz distance}:
\[\dbl(\mu,\nu) := \sup_{f\in \cF}\left(\int f\dd\mu - \int f\dd\nu\right).\]
One can easily check that this is a metric. Moreover, in the special case where $\mu$ and $\nu$ are both probability measures, moving $f$ by a constant does not change the result of $\int f\dd\mu - \int f\dd\nu$. Therefore, for a bounded domain $\Omega$, we can always assume $f(x_0)=0$ for a fixed $x_0\in \Omega$, then  $\norm{f}_\infty\leq \diam(\Omega)$ when computing the supremum in \eqref{eq: KR duality}. This implies $\dbl$-metric is equivalent to the classical $W_1$-metric when $\mu,\nu$ are both probability measures on a bounded domain $\Omega$:
\begin{align}
    W_1(\mu,\nu) &= \sup_{{\mathrm{Lip}(f)\leq 1}}\left(\int f\dd\mu - \int f\dd\nu\right) = \sup_{f\in \cF}
    \left(\int f\dd\mu - \int f\dd\nu\right)= \dbl(\mu,\nu).
\end{align}

\section{Private signed measure mechanism (PSMM)}
\label{sec: psmm}

We will first introduce PSMM, which is an easier and more intuitive approach. The procedure of PSMM is formally described in Algorithm~\ref{alg: direct}.    Note that in the \textit{output} step of Algorithm~\ref{alg: direct}, the size of the synthetic data $m'$ depends on the rational approximation of the density function of $\hat{\nu}$, and we discuss the details here. Let $\hat v_1,\dots,\hat v_m$ be the weight of the probability measure $\hat{\nu}$ on  $y_1,\dots,y_m$, respectively. We can choose rational numbers $r_1,\dots, r_m$ such that $\max_{i\in [m]}| r_i-\hat{\nu}_i|$ is arbitrarily small. Let $m'$ be the least common multiple of the denominators of $r_1,\dots, r_m$, then we output the synthetic dataset $\hat{\YY}$ containing $m'r_i$ copies of $y_i$ for $i=1,\dots,m$.

Before analyzing the privacy and accuracy of PSMM, we introduce a useful complexity measure of a given function class, which quantifies the influence of the Laplacian noise on the function class.

\begin{algorithm}[!ht]			
\caption{Private Signed Measure Mechanism} \label{alg: direct}
\begin{algorithmic}
\State {\bf Input:}  
 true data $\XX=(x_1,\ldots,x_n) \in \Omega^n$,  partition $(\Omega_1,\dots,\Omega_m)$ of $\Omega$, privacy parameter $\e >0$.

\begin{description}
  \item[Compute the true counts] Compute the true count in each regime  $n_i = \#\{x_j\in \Omega_i:j\in [n]\}$.
  \item[Create a new dataset] Choose any element $y_i\in\Omega_i$  {  independently of $\mathcal X$}, and let  $\cY$ be the collection of $n_i$ copies of $y_i$ for each $i\in [n]$.
  \item[Add noise]  Perturb the empirical measure $\mu_\cY$ of $\cY$  and obtain a signed measure $\nu$ such that 
  \[\nu (\{y_i\}) := (n_i + \lambda_i)/n,\]
  where  $\lambda_i\sim \Lap_{\mathbb Z}(1/\e)$ are i.i.d. discrete Laplacian random variables. 
  \item[Linear programming] Find the closest probability measure $\hat{\nu}$ of $\nu$ in $\dbl$-metric using Algorithm~\ref{alg:linear_programming}, and generate synthetic data $\widehat{\cY}$ from $\hat{\nu}$.
\end{description}

\State {\bf Output:} synthetic data 
  $\widehat{\cY} = (y_1,\ldots,y_{m'}) \in \Omega^{m'}$ for some integer $m'$.
\end{algorithmic}
\end{algorithm}

\begin{algorithm}[!ht]			
\caption{Linear Programming} \label{alg:linear_programming}
\begin{algorithmic}
\State {\bf Input:}  
 A discrete signed measure $\nu$ supported on $\cY=\{y_1,\dots,y_m\}$.
\begin{description}
    \item[Compute the distances] Compute the pairwise distances $\{\|y_i-y_j\|_{\infty}, i>j\}$.
    \item[Solve the linear programming] Solve the linear programming problem with $2m^2$ variables and $m+1$ constraints:
    \begin{align*}
        \min \quad & \sum_{i,j=1}^m  \|y_i - y_j\|_\infty(u_{ij}+u_{ij}') + 2\, v_i\\
        \mbox{s.t.}\quad & \sum_{j=1}^m (u_{ij} - u_{ij}') + v_i +\tau_i \geq \nu(\{y_i\}), & \forall {i\leq m},\\
        & \sum_{i=1}^m \tau_i = 1, &\\ 
        & u_{ij},u_{ij}',v_i, \tau_i \geq 0,  & \forall i,j\leq m, i\neq j.
    \end{align*}
\end{description}

\State {\bf Output:} a probability measure $\hat{\nu}$ with $\hat\nu(\{y_i\})=\tau_i$.

\end{algorithmic}
\end{algorithm}

\subsection{Laplacian complexity}
Given the Kantorovich-Rubinstein duality \eqref{eq: KR duality},  to control the $W_1$-distance between the original measure and the private measure, we need to describe how Lipschitz functions behave under Laplacian noise. As an analog of the worst-case Rademacher complexity     \cite{bartlett2002rademacher,foster2019vector}, we consider the worst-case Laplacian complexity. 
 {Such a worst-case complexity measure appears  since the original dataset is deterministic  without any distribution assumption.}

\begin{definition}[Worst-case Laplacian complexity]
    Let $\mathcal F$ be a function class on a metric space $\Omega$.  The worst-case Laplacian complexity of $\mathcal F$  is defined by
    \begin{align}
    L_n(\mathcal F):= \sup_{X_1,\dots, X_n\in \Omega}\mathbb E \left[\sup_{f\in \mathcal F}\left|\frac{1}{n}\sum_{i=1}^n\lambda_i f(X_i)\right|\right],
    \end{align}
    where $\lambda_1,\dots, \lambda_n \sim \Lap(1)$ are i.i.d. random variables.
\end{definition}

 {Since Laplacian random variables are sub-exponential but not sub-gaussian, its complexity measure is not equivalent to the Gaussian or Rademacher complexity, but it is related to the suprema of the mixed tail process \cite{dirksen2015tail} and the quadratic empirical process \cite{mendelson2010empirical}.}
Our next proposition bounds  $L_n(\mathcal F)$ in terms of the covering numbers of $\mathcal F$.  Its proof is  a  classical application of Dudley's chaining method (see, e.g., \cite{vershynin2018high}).

\begin{proposition}[Bounding Laplacian complexity with Dudley's entropy integral]
    \label{prop:Ln(f)}
    Suppose that $(\Omega,\rho)$ is a metric space and $\cF$ is a set of  functions on $\Omega$. Then  
    \[L_{n}(\cF)\leq C\inf_{\alpha> 0}\left(2\alpha + \frac{1}{\sqrt{n}}\int_{\alpha}^{\infty} \sqrt{{\log \cN(\cF,u,\|\cdot\|_{\infty})}}\dd u+\frac{1}{{n}}\int_{\alpha}^{\infty} {\log \cN(\cF,u,\|\cdot\|_{\infty})}\dd u\right)\]
    where $\cN(\cF,u,\|\cdot\|_{\infty})$ is the covering number of $\cF$ and $C>0$ is an absolute constant.
\end{proposition}

In particular, we are interested in the case where $\mathcal F$ is the class of all the bounded Lipschitz functions. One can find the result in \cite{tikhomirov1993varepsilon} or more explicit bound in \cite{gottlieb2016adaptive} that for the set $\cF$ of functions $f$ with $\norm{f}_\Lip\leq 1$, its covering number satisfies
\[\cN(\cF,u,\|\cdot\|_{\infty})\leq \left(\frac{8}{u}\right)^{\cN(\Omega, {u}/{2},\rho)}.\]

When $\Omega=[0,1]^d$, a better bound on the covering number for Lipschitz functions is available from \cite{tikhomirov1993varepsilon, von2004distance}:
    \[\cN(\cF,u,\|\cdot\|_{\infty})\leq \left(2\left\lceil{2}/u \right\rceil+1\right) 2^{\cN\left([0,1]^d, u/2, \|\cdot\|_\infty\right)},\]
    which implies the following corollary.

\begin{corollary}[Laplacian complexity for Lipschitz functions on the hypercube]
\label{cor: hypercube}
    Let $\Omega = [0,1]^d$ with the $\|\cdot\|_\infty$ metric, and $\cF$ be the set of all Lipschitz functions $f$ on $\Omega$ with $\norm{f}_\Lip\leq 1$. We have
        \[L_{n}(\mathcal F)  \leq \left\{\begin{aligned}
        & C n^{-1/2}\quad &\textrm{ if } d=1,\\
         &{C\log n} \cdot n^{-1/2}\quad &\textrm{ if } d=2,\\
        &Cd^{-1}  {n^{-{1}/{d}}}\quad &\textrm{ if } d\geq 3.
        \end{aligned}\right.\]
\end{corollary}

\paragraph{Discrete Laplacian complexity}
Laplacian complexity can be  useful  for  differential privacy algorithms based on the Laplacian mechanism \cite{dwork2014algorithmic}. However, since PSMM perturbs  counts in each subregion, it is more convenient for us to add integer noise to the true counts. Instead, we will use the \textit{worst-case discrete Laplacian complexity} defined below:
 \begin{align}
   \widetilde{L}_n(\mathcal F):= \sup_{X_1,\dots, X_n\in \Omega}\mathbb E \left[\sup_{f\in \mathcal F}\left|\frac{1}{n}\sum_{i=1}^n\lambda_i f(X_i)\right|\right],
    \end{align}
    where $\lambda_1,\dots, \lambda_n \sim \Lap_{\mathbb Z}(1)$ are i.i.d. discrete Laplacian random variables.
   
    In particular, $\Lap_{\mathbb Z}(1)$ has a bounded sub-exponential norm,    therefore the proof of Proposition~\ref{prop:Ln(f)} works for discrete Laplacian random variables as well. Consequently,   Corollary~\ref{cor: hypercube} also holds for $\widetilde{L}_n(\mathcal F)$, with a different absolute constant $C$. 

\subsection{Privacy and Accuracy of Algorithm \ref{alg: direct}}

The privacy guarantee of Algorithm~\ref{alg: direct} can be proved by checking the definition. The essence of the proof is the same as the classical Laplacian mechanism \cite{dwork2014algorithmic}.
\begin{proposition}[Privacy of Algorithm~\ref{alg: direct}]\label{prop: direct_priv}
	 Algorithm~\ref{alg: direct} is $\e$-differentially private.
\end{proposition}

 We now turn to accuracy.  The linear programming problem stated in Algorithm \ref{alg:linear_programming} has  $(2m^2+2m)$ many variables and $(m+1)$ many constraints, which can be solved in polynomial time in $m$. We first show that Algorithm \ref{alg:linear_programming} indeed outputs the closest probability measure to $\nu$ in the $\dbl$-distance in the next proposition.

\begin{proposition}\label{prop:LP}
    For a discrete signed measure $\nu$ on $\Omega$, Algorithm~\ref{alg:linear_programming} gives its closest probability measure in $\dbl$-distance with the same support set with a polynomial running time in $m$.
\end{proposition}

Now we are ready to analyze the accuracy of Algorithm~\ref{alg: direct}. {In PSMM,  independent Laplacian noise is added to the count of each sub-region. Therefore, the Laplacian complexity arises when considering the expected Wasserstein distance between the original empirical measure and the synthetic measure.}

\begin{theorem}[Accuracy of Algorithm~\ref{alg: direct}]
    \label{thm: accuracy_direct}
    Suppose $(\Omega_1,\dots,\Omega_m)$ is a partition of $(\Omega,\rho)$ and $\mathcal F$ is the set of all functions with Lipschitz norm bounded by $1$. Then the measure $\hat\nu$ generated from Algorithm~\ref{alg: direct} satisfies
    \[\E W_1(\mu_\XX,\hat\nu)\leq \max_i \diam(\Omega_i) + \frac{2 m}{\e n} \widetilde L_m(\cF).\]
\end{theorem}

Note that $\diam(\Omega_i)\asymp m^{-1/d}$ can be satisfied when we take a  partition of $\Omega=[0,1]^d$ where each $\Omega_i$ is a subcube of the same size. Using the formula above and the result of Laplacian complexity for the hypercube in Corollary~\ref{cor: hypercube}, one can easily deduce the following result.  

\begin{corollary}[Accuracy of Algorithm~\ref{alg: direct} on the hypercube]\label{cor: accuracy_direct_hypercube}
Take $m=\lceil\e n\rceil$ and let $(\Omega_1,\dots,\Omega_m)$ be a partition of $\Omega=[0,1]^d$ with the norm $\|\cdot\|_\infty$. Assume that   $\diam(\Omega_i)\asymp m^{-1/d}$. Then the measure $\hat{\nu}$ generated from Algorithm~\ref{alg: direct} satisfies
\[\E W_1(\mu_\cX,\hat{\nu}) \leq \left\{\begin{aligned}
    &C(\e n)^{-\frac{1}{2}}&\quad \textrm{ if } d=1,\\
    &{C}\log(\e n)(\e n)^{-\frac{1}{2}}&\quad \textrm{ if } d=2,\\
    &C {(\e n)^{-\frac{1}{d}}}&\quad \textrm{ if } d\geq 3.
\end{aligned}\right.\]
\end{corollary}

\section{Private measure mechanism (PMM)}
\label{sec: pmm}

\subsection{Binary partition and noisy counts}

A binary hierarchical partition of a set $\Omega$ of depth $r$ 
is a family of subsets $\Omega_\theta$ indexed by 
$\theta \in \{0,1\}^{\le r}$, where 
$$
\{0,1\}^{\le k} = \{0,1\}^0 \sqcup \{0,1\}^1 \sqcup \cdots \sqcup \{0,1\}^k, 
\quad k=0,1,2\dots, 
$$
and such that $\Omega_{\theta}$ is partitioned into $\Omega_{\theta0}$ and $\Omega_{\theta 1}$ for every $\theta \in \{0,1\}^{\le r-1}$.
By convention, the cube $\{0,1\}^0$ consists of a single element $\emptyset$. We usually drop the subscript $\emptyset$ and write $n$ instead of $n_\emptyset$. When $\theta \in \{0,1\}^j$, we call $j$ the {\em level} of $\theta$. We can also encode a binary hierarchical partition of $\Omega$ in a binary tree of depth $r$, where the root is labeled $\Omega$ and the $j$-th level of the tree encodes the subsets $\Omega_{\theta}$ for $\theta$ at level $j$.


Let  $(\Omega_\theta)_{\theta \in \{0,1\}^{\le r}}$ be a binary partition of $\Omega$. 
Given true data $(x_1,\ldots,x_n) \in \Omega^n$, the {\em true count} $n_\theta$ 
is the number of data points in the region $\Omega_\theta$, i.e. 
$$
n_\theta \coloneqq \abs{\left\{ i \in [n]: \; x_i \in \Omega_\theta \right\}}.
$$
We will convert true counts into {\em noisy counts} $n'_\theta$ by adding Laplacian noise;
all regions on the same level will receive noise of the same expected magnitude. 
Formally, we set
$$
n'_\theta \coloneqq \left( n_\theta + \l_\theta \right)_+, 
\quad \text{where} \quad
\l_\theta \sim \Lap_{\mathbb Z}(\s_j),
$$
and $j \in \{0,\ldots,r\}$ is the level of $\theta$. At this point, the magnitudes of the noise $\s_j$ can be arbitrary.

\subsection{Consistency}

The true counts $n_\theta$ are non-negative and {\em consistent}, i.e., the counts of subregions always add up to the count of the region:
$$
n_{\theta0}+n_{\theta1}=n_\theta
\quad \text{for all } \theta \in \{0,1\}^{\le r-1}.
$$
The noisy counts $n'_\theta$ are non-negative, but not necessarily consistent.
Algorithm~\ref{alg: consistency} enforces consistency by adjusting the counts iteratively, from top to bottom. In the case of the deficit, when the sum of the two subregional counts is smaller than the count of the region, we increase both subregional counts. In the opposite case or surplus, we decrease both subregional counts. Apart from this requirement, we are free to distribute the deficit or surplus between the subregional counts.  

It is convenient to state this requirement by considering a {\em product partial order} on $\Z_+^2$, where we declare that $(a_0,a_1) \preceq (b_0,b_1)$ if and only if $a_0 \le b_0$ and $a_1 \le b_1$. We call the two vectors $a,b \in \Z^2$ {\em comparable} if either $a \preceq b$ or $b \preceq a$.
Furthermore, $L(a)$ denotes the line $x+y=a$ on the plane. 

\begin{algorithm}[!ht]			
\caption{Consistency} \label{alg: consistency}
\begin{algorithmic} 

\State {\bf Input:} non-negative numbers $(n'_\theta)_{\theta \in \{0,1\}^{\le r}}$, where $n'$ is a nonnegative integer.
\State set $m \coloneqq n'$.

\For{$j=0,\ldots,r-1$}
    \For{$\theta \in \{0,1\}^j$}
        \State transform the vector $(n'_{\theta0},n'_{\theta1}) \in \Z_+^2$ into any comparable vector $(m_{\theta0},m_{\theta1}) \in \Z_+^2 \cap L(m_\theta)$. 	
    \EndFor
\EndFor

\State {\bf Output:} non-negative integers $(m_\theta)_{\theta \in \{0,1\}^{\le r}}$.

\end{algorithmic}
\end{algorithm}

At each step, Algorithm~\ref{alg: consistency} uses a transformation 
$
f_\theta: \Z_+^2 \to \Z_+^2 \cap L(m_\theta).
$
It can be chosen arbitrarily; the only requirement is that $f_\theta(x)$ be comparable with $x$.  The comparability requirement is natural and non-restrictive.
For example, the {\em uniform} transformation selects the closest point in the discrete interval $\Z_+^2 \cap L(m_\theta)$ in (say) the Euclidean metric. Alternatively, the {\em proportional} transformation selects the point in the discrete interval $\Z_+^2 \cap L(m_\theta)$ that is closest to the line that connects the input vector and the origin.

\subsection{Synthetic data}

Algorithm~\ref{alg: consistency} ensures that the output counts $m_\theta$ are non-negative, integer, and consistent. They are also private since they are a function of the noisy counts $n'_\theta$, which are private as we proved.  Therefore, the counts $m_\theta$ can be used to generate {\em private synthetic data} by putting $m_\theta$ points in cell $\Omega_\theta$. Algorithm~\ref{alg: synthetic data} makes this formal.

\begin{algorithm}[!ht]			
\caption{Private Measure Mechanism} \label{alg: synthetic data}
\begin{algorithmic}

\State {\bf Input:}   true data $\XX=(x_1,\ldots,x_n) \in \Omega^n$, noise magnitudes $\s_0,\ldots,\s_r > 0$.

\begin{description}
  \item[Compute true counts] Let $n_\theta$ be the number of data points in $\Omega_\theta$.
  \item[Add noise] Let 
  $n'_\theta \coloneqq (n_\theta + \l_\theta)_+$, where $\l_\theta \sim \Lap_{\Z}(\s_j)$ are i.i.d. random variables, 
  \item[Enforce consistency] Convert the noisy counts $(n'_\theta)$ to consistent counts $(m_\theta)$ using Algorithm~\ref{alg: consistency}.
  \item[Sample] Choose any $m_\theta$ points in each cell $\Omega_\theta$, $\theta \in \{0,1\}^r$ {independently of $\mathcal X$.}
\end{description}

\State {\bf Output:} the set of all these points as synthetic data 
  $\YY = (y_1,\ldots,y_m) \in \Omega^m$.

\end{algorithmic}
\end{algorithm}


\subsection{Privacy and accuracy  of Algorithm \ref{alg: synthetic data}}
\label{sec: analysis of pmm}
 
We first prove that Algorithm~\ref{alg: synthetic data} is differentially private. The proof idea is similar to the classic Laplacian mechanism. But now our noise is of differential scale for each level, so more delicate calculations are needed.

\begin{theorem}[Privacy of Algorithm~\ref{alg: synthetic data}]		\label{thm: privacy}
  The vector of noisy counts $(n_\theta + \l_\theta)$ in Algorithm~\ref{alg: synthetic data}
  is $\e$-differentially private, where
  $$
  \e = \sum_{j=0}^r \frac{1}{\s_j}.
  $$
  Consequently, the synthetic data $\YY$ generated by Algorithm~\ref{alg: synthetic data} 
  is $\e$-differentially private.
\end{theorem}


Having analyzed the privacy of the synthetic data, we now turn to its accuracy.
It is determined by the magnitudes of the noise $\s_j$ and by 
the multiscale geometry of the domain $\Omega$. The latter is captured
by the diameters of the regions $\Omega_\theta$, specifically 
by their sum at each level, which we denote
\begin{equation}	\label{eq: Deltaj}
\Delta_j \coloneqq \sum_{\theta \in \{0,1\}^j} \diam(\Omega_\theta)
\end{equation}
and adopt the notation
$
\Delta_{-1} \coloneqq 
\Delta_0
= \diam(\Omega).
$
In addition to $\Delta_j$, 
the accuracy is affected by the {\em resolution} of the partition, 
which is the maximum diameter of the cells, denoted by
$$
\d \coloneqq \max_{\theta \in \{0,1\}^r} \diam(\Omega_\theta).
$$

\begin{theorem}[Accuracy of Algorithm~\ref{alg: synthetic data}]			\label{thm: accuracy}
  Algorithm~\ref{alg: synthetic data} that transforms true data $\XX$ into synthetic data $\YY$
  has the following expected accuracy in the Wasserstein metric:
  $$
  \E W_1 \left( \mu_\XX, \mu_\YY \right) 
  \le \frac{2\sqrt{2}}{n} \sum_{j=0}^r \s_j \Delta_{j-1}+\d.
  $$
  Here $\mu_\XX$ and $\mu_\YY$ are the empirical probability distributions
  on the true and synthetic data, respectively.
\end{theorem}


The privacy and accuracy guarantees of Algorithm~\ref{alg: synthetic data} (Theorems~\ref{thm: privacy} and \ref{thm: accuracy}) hold for any choice of noise levels $\s_j$. By optimizing $\s_j$, we can achieve the best accuracy for a given level of privacy. 

\begin{theorem}[Optimized accuracy]			\label{thm: optimized accuracy}
  With the optimal choice of magnitude levels \eqref{eq: optimal noise}, 
  Algorithm~\ref{alg: synthetic data} that transforms true data $\XX$ into synthetic data $\YY$
  is $\e$-differential private, and 
  has the following expected accuracy in the $1$-Wasserstein distance:
  $$
  \E W_1 \left( \mu_\XX, \mu_\YY \right) 
  \le \frac{\sqrt{2}}{\e n} \Big( \sum_{j=0}^r \sqrt{\Delta_{j-1}} \Big)^2+ \d.
  $$
  Here $\mu_\XX$ and $\mu_\YY$ are the empirical measures
  of the true and synthetic data, respectively.
\end{theorem}

\begin{corollary}[Optimized accuracy for hypercubes] \label{cor: PMM_hypercube}
When $\Omega=[0,1]^d$ equipped with the $\ell^{\infty}$ metric, with the optimal choice of magnitude levels \eqref{eq: optimal noise} and the optimal choice of 
\[r =\left\{\begin{aligned}
    &\log_2(\e n)-1 &\textrm{ if } d=1,\\
    &\log_2(\e n) &\textrm{ if }d\geq 2,
\end{aligned}
\right. \]
we have
\[\E W_1(\mu_\XX,\mu_\YY)\lesssim \left\{\begin{aligned}
    \frac{\log^2(\e n)}{\e n}, \quad \textrm{ if }d=1,\\
    (\e n)^{-1/d}, \quad \textrm{ if }d\geq 2.
\end{aligned}\right.\]
    
\end{corollary}

\begin{remark}[Computational efficiency of Algorithm~\ref{alg: synthetic data}]
\label{rmk: PMM_running_time}
Since a binary hierarchical partition has $2^r$ cells in total, the running time of Algorithm~\ref{alg: synthetic data} is $O(2^r)$. When $\Omega=[0,1]^d$, with the same optimal choice of $r$ in  Corollary~\ref{cor: PMM_hypercube}, the running time of PMM becomes $O(\e dn)$.
\end{remark}

\subsection{Proof of Theorem~\ref{thm: accuracy}}
For the proof of Theorem~\ref{thm: accuracy}, we introduce a quantitative notion for the incomparability 
of two vectors on the plane.  For vectors $a,b \in \Z_+^2$, we define 
  $$
  \flux(a,b):=
  \begin{cases}
    0 & \text{if $a$ and $b$ are comparable,}\\
    \min \left( \,\abs{a_1-b_1},\, \abs{a_2-b_2} \right) & \text{otherwise}.
  \end{cases}
  $$

\begin{lemma}[Flux as incomparability]		\label{lem: flux incomparability}
  $\flux(a,b)$ is the $\ell_\infty$-distance from $a$ 
  to the set of points that are comparable to $b$.
\end{lemma} 

For example, if $a=(1,9)$ and $b=(6,7)$, then $\flux(a,b)=2$. 
Note that $a$ has a distance $2$ to the vector $(1,7)$ which is comparable with $b$.

\begin{lemma}[Flux as transfer]		\label{lem: flux transfer}
  Suppose we have two bins with $a_1$ and $a_2$ balls in them.
  Then one can achieve $b_1$ and $b_2$ balls in these bins by:
  \begin{enumerate}[(a)]
    \item first making the total number of balls correct by adding a total of $(b_1+b_2)-(a_1-a_2)$ balls to the two bins (or removing, if that number is negative);
    \item then transferring $\flux \left( (a_1,a_2), \, (b_1,b_2) \right)$ balls from one bin to the other.
  \end{enumerate}
\end{lemma}

For example, suppose that one bin has $1$ ball and the other has $9$. Then we can achieve $6$ and $7$ balls in these bins by first adding $3$ balls to the first bin and transferring $2$ balls from the second to the first bin. As we noted above, $2$ is the flux between the vectors $(1,9)$ and $b=(6,7)$. 

Lemma~\ref{lem: flux transfer} can be generalized to the hierarchical binary partition of $\Omega$ as follows. 
\begin{lemma} \label{lem: tranform data}
  Consider any data set $\XX \in \Omega^n$, 
  and let $(n_\theta)_{\theta \in \{0,1\}^r}$ be its counts.  
  Consider any consistent vector of non-negative integers $(m_\theta)_{\theta \in \{0,1\}^r}$.
  Then one can transform $\XX$ into a set $\ZZ \in \Omega^m$ that has counts $(m_\theta)_{\theta \in \{0,1\}^r}$ by:
  \begin{enumerate}[(a)]
    \item first making the total number of points correct by adding a total of $m-n$ points to $\Omega$ (or remove, if that number is negative);
    \item then transferring $\flux \left( (n_{\theta0},n_{\theta1}), \, (m_{\theta0},m_{\theta1}) \right)$ points from $\Omega_{\theta0}$ to $\Omega_{\theta1}$ or vice versa, for all $j=0,\ldots,r-1$ and $\theta \in \{0,1\}^j$. 
  \end{enumerate}  
\end{lemma}

Combining the concept of the $\flux$ and our algorithm, the following two lemmas are  useful in the proof of Theorem~\ref{thm: accuracy}.
\begin{lemma}       \label{lem: flux and laplacian}
  In Algorithm~\ref{alg: synthetic data}, we have
  $$
  \flux \left( (n_{\theta0},n_{\theta1}), \, (m_{\theta0},m_{\theta1}) \right)
  \le \max \left(\,\abs{\l_{\theta0}}, \abs{\l_{\theta1}} \right)
  $$
  for all $j=0,\ldots,r-1$ and $\theta \in \{0,1\}^j$. 
\end{lemma}

\begin{lemma}		\label{lem: adding points}
  For any finite multisets $U \subset V$ such that all elements in $U$ are from $\Omega$, one has
  $$
  W_1(\mu_U, \mu_V) \le \frac{\abs{V \setminus U}}{\abs{V}} \cdot \diam(\Omega).
  $$
\end{lemma}

\begin{proof}\emph{(Proof of Theorem~\ref{thm: accuracy})}
Owing to Lemma~\ref{lem: tranform data} and Lemma~\ref{lem: flux and laplacian}, the creation of synthetic data from the true data $\XX \mapsto \YY$, described by Algorithm~\ref{alg: synthetic data}, can be achieved by the following three steps.
\begin{enumerate}[1.]
  \item \label{st: add/remove} Transform the $n$-point input set $\XX$ to an $m$-point set $\XX_1$ by adding or removing $\abs{m-n}$ points.
  \item \label{st: move} Transform $\XX_1$ to $\XX_2$ by moving at most $\max \left( \,\abs{\l_{\theta0}}, \abs{\l_{\theta1}} \right)$ many data points for each $j=0,1,\dots,r-1$ and $\theta \in \{0,1\}^{j}$ 
between the two parts of the region $\Omega_\theta$. 
  \item \label{st: shake} Transforms $\XX_2$ to the output data $\YY$ by relocating points within their cells. 
\end{enumerate} 

We will analyze the accuracy of these steps one at a time.

{\em Analyzing Step~\ref{st: move}. }
The total distance the points are moved at this step is bounded by 
\begin{equation}	\label{eq: D}
\sum_{j=0}^{r-1}\sum_{\theta \in \{0,1\}^{j}} \max \left( \,\abs{\l_{\theta0}}, \abs{\l_{\theta1}} \right) \diam(\Omega_\theta)
\eqqcolon D.
\end{equation}
Since $\abs{\XX_1} = m$, it follows that 
\begin{equation}	\label{eq: X1 X2}
W_1(\mu_{\XX_1}, \mu_{\XX_2}) \le \frac{D}{m}.
\end{equation}

{\em Combining Steps \ref{st: add/remove} and \ref{st: move}. }
Recall that step~\ref{st: add/remove} transforms the input data $\XX$ with 
$\abs{\XX}=n$ into $\XX_1$ with $\abs{\XX_1}=m =n+\sign(\lambda)\cdot \lfloor|\lambda|\rfloor$ by adding or removing points, 
depending on the sign of $\l$.

{\sl Case 1: $\l \ge 0$.} Here $\XX_1$ is obtained from $\XX$ by adding $\lfloor\l\rfloor$ points, so Lemma~\ref{lem: adding points} gives
$$
W_1(\mu_\XX, \mu_{\XX_1}) 
\le \frac{\l}{m} \cdot \Delta_0.
$$
Combining this with \eqref{eq: X1 X2} by triangle inequality, we conclude that 
$$
W_1(\mu_\XX, \mu_{\XX_2}) 
\le \frac{\l \Delta_0 + D}{m}
\le  \frac{\l \Delta_0 + D}{n}.
$$

{\sl Case 2: $\l < 0$.} Here $\XX_1$ is obtained from $\XX$ by removing a set $\XX_0$ of $n-m= \lfloor\abs{\l}\rfloor$ points. Furthermore, by our analysis of step~\ref{st: move}, $\XX_2$ is obtained from $\XX_1$ by moving points the total distance at most $D$. Therefore, $\XX_2 \cup \XX_0$ (as a multiset) is obtained from $\XX = \XX_1 \cup \XX_0$ by moving points the total distance at most $D$, too. (The points in $\XX_0$ remain unmoved.) Since $\abs{\XX} = n$, it follows that 
$$
W_1(\mu_\XX, \mu_{\XX_2 \cup \XX_0}) \le \frac{D}{n}.
$$
Moreover, Lemma~\ref{lem: adding points} gives
$$
W_1(\mu_{\XX_2}, \mu_{\XX_2 \cup \XX_0}) 
\le \frac{\abs{\XX_0}}{\abs{\XX_2 \cup \XX_0}} \cdot \diam(\Omega)
\le \frac{\abs{\l} \Delta_0}{n}.
$$
(Here we used that the multiset $\XX_2 \cup \XX_0$ has the same number of points as $\XX$, which is $n$.)
Combining the two bounds by triangle inequality, we obtain 
\begin{equation}	\label{eq: X X2}
W_1(\mu_\XX, \mu_{\XX_2}) 
\le \frac{\abs{\l} \Delta_0 + D}{n}.
\end{equation}
In other words, this bound holds in both cases.

{\em Analyzing Step 3. } 
This step is the easiest to analyze: since $\YY$ is obtained from $\XX_2$ by relocating the points are relocated within their cells, and the maximal diameter of the cells is $\d$, we have
$
W_1(\mu_{\XX_2},\mu_\YY) \le \d.
$
Combining this with \eqref{eq: X X2} by triangle inequality, we conclude that
$$
W_1(\mu_\XX, \mu_\YY) 
\le \frac{\abs{\l} \Delta_0 + D}{n} + \d.
$$

{\em Taking expectation. } 
Recall the definition of $D$ from \eqref{eq: D}. We get
$$
\E W_1(\mu_\XX, \mu_\YY) 
\le \frac{1}{n} \left[\E \left[\,\abs{\l}\right] \Delta_0
	+ \sum_{j=0}^{r-1} \sum_{\theta \in \{0,1\}^j} \E \left[ \max \left( \,\abs{\l_{\theta0}}, \abs{\l_{\theta1}} \right) \right]
		 \diam(\Omega_\theta)
 \right] + \d.
$$
Since $\l \sim \Lap_\Z(\s_0)$, by \eqref{eq: Lap variance} we have $\E \left[\,\abs{\l}\right] \le (\E(\l)^2)^{1/2} \le \sqrt{2} \s_0.$ Similarly, since $\l_{\theta0}$ and $\l_{\theta1}$ are independent $\Lap_\Z(\s_{j+1})$ random variables,  
$\E \left[ \max \left(\, \abs{\l_{\theta0}}, \abs{\l_{\theta1}} \right) \right] \le 2\sqrt{2} \s_{j+1}$. 
Substituting these estimates and rearranging the terms of the sum will complete the proof. 
\end{proof}

 \section*{Acknowledgements}
R.V. acknowledges support from NSF DMS--1954233, NSF DMS--2027299, U.S. Army 76649--CS, and NSF-Simons Research Collaborations on the Mathematical and Scientific Foundations of Deep Learning. Y.Z. is partially supported by NSF-Simons Research Collaborations on the Mathematical and Scientific Foundations of Deep Learning.

The authors thank March Boedihardjo, Girish Kumar, and Thomas Strohmer for helpful discussions.

\bibliographystyle{plain}
\bibliography{ref}

\begin{thebibliography}{10}

\bibitem{abadi2016deep}
Martin Abadi, Andy Chu, Ian Goodfellow, H~Brendan McMahan, Ilya Mironov, Kunal
  Talwar, and Li~Zhang.
\newblock Deep learning with differential privacy.
\newblock In {\em Proceedings of the 2016 ACM SIGSAC conference on computer and
  communications security}, pages 308--318, 2016.

\bibitem{abowd2019census}
John Abowd, Robert Ashmead, Garfinkel Simson, Daniel Kifer, Philip Leclerc,
  Ashwin Machanavajjhala, and William Sexton.
\newblock Census topdown: Differentially private data, incremental schemas, and
  consistency with public knowledge.
\newblock {\em US Census Bureau}, 2019.

\bibitem{abowd20222020}
John~M Abowd, Robert Ashmead, Ryan Cumings-Menon, Simson Garfinkel, Micah
  Heineck, Christine Heiss, Robert Johns, Daniel Kifer, Philip Leclerc, Ashwin
  Machanavajjhala, et~al.
\newblock The 2020 census disclosure avoidance system topdown algorithm.
\newblock {\em Harvard Data Science Review}, (Special Issue 2), 2022.

\bibitem{ba2011sublinear}
Khanh~Do Ba, Huy~L Nguyen, Huy~N Nguyen, and Ronitt Rubinfeld.
\newblock Sublinear time algorithms for earth mover’s distance.
\newblock {\em Theory of Computing Systems}, 48:428--442, 2011.

\bibitem{barak2007privacy}
Boaz Barak, Kamalika Chaudhuri, Cynthia Dwork, Satyen Kale, Frank McSherry, and
  Kunal Talwar.
\newblock Privacy, accuracy, and consistency too: a holistic solution to
  contingency table release.
\newblock In {\em Proceedings of the twenty-sixth ACM SIGMOD-SIGACT-SIGART
  symposium on Principles of database systems}, pages 273--282, 2007.

\bibitem{bartlett2002rademacher}
Peter~L Bartlett and Shahar Mendelson.
\newblock {R}ademacher and {G}aussian complexities: Risk bounds and structural
  results.
\newblock {\em Journal of Machine Learning Research}, 3(Nov):463--482, 2002.

\bibitem{bellovin2019privacy}
Steven~M Bellovin, Preetam~K Dutta, and Nathan Reitinger.
\newblock Privacy and synthetic datasets.
\newblock {\em Stan. Tech. L. Rev.}, 22:1, 2019.

\bibitem{blum2013learning}
Avrim Blum, Katrina Ligett, and Aaron Roth.
\newblock A learning theory approach to noninteractive database privacy.
\newblock {\em Journal of the ACM (JACM)}, 60(2):1--25, 2013.

\bibitem{bobkov2021simple}
Sergey~G Bobkov and Michel Ledoux.
\newblock A simple fourier analytic proof of the {AKT} optimal matching
  theorem.
\newblock {\em The Annals of Applied Probability}, 31(6):2567--2584, 2021.

\bibitem{boedihardjo2022covariance}
March Boedihardjo, Thomas Strohmer, and Roman Vershynin.
\newblock Covariance’s loss is privacy’s gain: Computationally efficient,
  private and accurate synthetic data.
\newblock {\em Foundations of Computational Mathematics}, pages 1--48, 2022.

\bibitem{boedihardjo2022privacy}
March Boedihardjo, Thomas Strohmer, and Roman Vershynin.
\newblock Privacy of synthetic data: A statistical framework.
\newblock {\em IEEE Transactions on Information Theory}, 69(1):520--527, 2022.

\bibitem{boedihardjo2022private}
March Boedihardjo, Thomas Strohmer, and Roman Vershynin.
\newblock Private measures, random walks, and synthetic data.
\newblock {\em arXiv preprint arXiv:2204.09167}, 2022.

\bibitem{boedihardjo2022private2}
March Boedihardjo, Thomas Strohmer, and Roman Vershynin.
\newblock Private sampling: a noiseless approach for generating differentially
  private synthetic data.
\newblock {\em SIAM Journal on Mathematics of Data Science}, 4(3):1082--1115,
  2022.

\bibitem{boedihardjo2023covariance}
March Boedihardjo, Thomas Strohmer, and Roman Vershynin.
\newblock Covariance loss, {S}zemeredi regularity, and differential privacy.
\newblock {\em arXiv preprint arXiv:2301.02705}, 2023.

\bibitem{bubeck2021universal}
S{\'e}bastien Bubeck and Mark Sellke.
\newblock A universal law of robustness via isoperimetry.
\newblock {\em Advances in Neural Information Processing Systems},
  34:28811--28822, 2021.

\bibitem{chan2011private}
T-H~Hubert Chan, Elaine Shi, and Dawn Song.
\newblock Private and continual release of statistics.
\newblock {\em ACM Transactions on Information and System Security (TISSEC)},
  14(3):1--24, 2011.

\bibitem{chaudhuri2008privacy}
Kamalika Chaudhuri and Claire Monteleoni.
\newblock Privacy-preserving logistic regression.
\newblock {\em Advances in neural information processing systems}, 21, 2008.

\bibitem{dereich2013constructive}
Steffen Dereich, Michael Scheutzow, and Reik Schottstedt.
\newblock Constructive quantization: Approximation by empirical measures.
\newblock {\em Annales de l'IHP Probabilit{\'e}s et statistiques},
  49(4):1183--1203, 2013.

\bibitem{dirksen2015tail}
Sjoerd Dirksen.
\newblock Tail bounds via generic chaining.
\newblock {\em Electron. J. Probab}, 20(53):1--29, 2015.

\bibitem{domingo2021limits}
Josep Domingo-Ferrer, David S{\'a}nchez, and Alberto Blanco-Justicia.
\newblock The limits of differential privacy (and its misuse in data release
  and machine learning).
\newblock {\em Communications of the ACM}, 64(7):33--35, 2021.

\bibitem{duchi2018minimax}
John~C Duchi, Michael~I Jordan, and Martin~J Wainwright.
\newblock Minimax optimal procedures for locally private estimation.
\newblock {\em Journal of the American Statistical Association},
  113(521):182--201, 2018.

\bibitem{dwork2010differential}
Cynthia Dwork, Moni Naor, Toniann Pitassi, and Guy~N Rothblum.
\newblock Differential privacy under continual observation.
\newblock In {\em Proceedings of the forty-second ACM symposium on Theory of
  computing}, pages 715--724, 2010.

\bibitem{dwork2015efficient}
Cynthia Dwork, Aleksandar Nikolov, and Kunal Talwar.
\newblock Efficient algorithms for privately releasing marginals via convex
  relaxations.
\newblock {\em Discrete \& Computational Geometry}, 53:650--673, 2015.

\bibitem{dwork2014algorithmic}
Cynthia Dwork and Aaron Roth.
\newblock The algorithmic foundations of differential privacy.
\newblock {\em Foundations and Trends{\textregistered} in Theoretical Computer
  Science}, 9(3--4):211--407, 2014.

\bibitem{foster2019vector}
Dylan~J Foster and Alexander Rakhlin.
\newblock $\ell^{\infty}$ vector contraction for {R}ademacher complexity.
\newblock {\em arXiv preprint arXiv:1911.06468}, 6, 2019.

\bibitem{gottlieb2016adaptive}
Lee-Ad Gottlieb, Aryeh Kontorovich, and Robert Krauthgamer.
\newblock Adaptive metric dimensionality reduction.
\newblock {\em Theoretical Computer Science}, 620:105--118, 2016.

\bibitem{hardt2012simple}
Moritz Hardt, Katrina Ligett, and Frank McSherry.
\newblock A simple and practical algorithm for differentially private data
  release.
\newblock {\em Advances in neural information processing systems}, 25, 2012.

\bibitem{hauer2021differential}
Mathew~E Hauer and Alexis~R Santos-Lozada.
\newblock Differential privacy in the 2020 census will distort covid-19 rates.
\newblock {\em Socius}, 7:2378023121994014, 2021.

\bibitem{hawes2020implementing}
Michael~B Hawes.
\newblock Implementing differential privacy: Seven lessons from the 2020
  {U}nited {S}tates {C}ensus.
\newblock {\em Harvard Data Science Review}, 2(2), 2020.

\bibitem{IK}
Seidu Inusah and Tomasz Kozubowski.
\newblock A discrete analogue of the {L}aplace distribution.
\newblock {\em Journal of Statistical Planning and Inference}, 136:1090--1102,
  03 2006.

\bibitem{jordon2019pate}
James Jordon, Jinsung Yoon, and Mihaela Van Der~Schaar.
\newblock {PATE-GAN}: Generating synthetic data with differential privacy
  guarantees.
\newblock In {\em International conference on learning representations}, 2019.

\bibitem{kovalev2022lipschitz}
Leonid~V Kovalev.
\newblock {L}ipschitz clustering in metric spaces.
\newblock {\em The Journal of Geometric Analysis}, 32(7):188, 2022.

\bibitem{liu2021iterative}
Terrance Liu, Giuseppe Vietri, and Steven~Z Wu.
\newblock Iterative methods for private synthetic data: Unifying framework and
  new methods.
\newblock {\em Advances in Neural Information Processing Systems}, 34:690--702,
  2021.

\bibitem{mendelson2010empirical}
Shahar Mendelson.
\newblock Empirical processes with a bounded $\psi$ 1 diameter.
\newblock {\em Geometric and Functional Analysis}, 20(4):988--1027, 2010.

\bibitem{meunier2022dynamical}
Laurent Meunier, Blaise~J Delattre, Alexandre Araujo, and Alexandre Allauzen.
\newblock A dynamical system perspective for {L}ipschitz neural networks.
\newblock In {\em International Conference on Machine Learning}, pages
  15484--15500. PMLR, 2022.

\bibitem{song2013stochastic}
Shuang Song, Kamalika Chaudhuri, and Anand~D Sarwate.
\newblock Stochastic gradient descent with differentially private updates.
\newblock In {\em 2013 IEEE global conference on signal and information
  processing}, pages 245--248. IEEE, 2013.

\bibitem{su2016differentially}
Dong Su, Jianneng Cao, Ninghui Li, Elisa Bertino, and Hongxia Jin.
\newblock Differentially private $k$-means clustering.
\newblock In {\em Proceedings of the sixth ACM conference on data and
  application security and privacy}, pages 26--37, 2016.

\bibitem{talagrand1992matching}
Michel Talagrand.
\newblock Matching random samples in many dimensions.
\newblock {\em The Annals of Applied Probability}, pages 846--856, 1992.

\bibitem{talagrand2005generic}
Michel Talagrand.
\newblock {\em The generic chaining: upper and lower bounds of stochastic
  processes}.
\newblock Springer Science \& Business Media, 2005.

\bibitem{thaler2012faster}
Justin Thaler, Jonathan Ullman, and Salil Vadhan.
\newblock Faster algorithms for privately releasing marginals.
\newblock In {\em Automata, Languages, and Programming: 39th International
  Colloquium, ICALP 2012, Warwick, UK, July 9-13, 2012, Proceedings, Part I
  39}, pages 810--821. Springer, 2012.

\bibitem{tikhomirov1993varepsilon}
VM~Tikhomirov.
\newblock $\varepsilon$-entropy and $\varepsilon$-capacity of sets in
  functional spaces.
\newblock In {\em Selected works of AN Kolmogorov}, pages 86--170. Springer,
  1993.

\bibitem{ullman2011pcps}
Jonathan Ullman and Salil Vadhan.
\newblock {PCP}s and the hardness of generating private synthetic data.
\newblock In {\em Theory of Cryptography: 8th Theory of Cryptography
  Conference, TCC 2011, Providence, RI, USA, March 28-30, 2011. Proceedings 8},
  pages 400--416. Springer, 2011.

\bibitem{vadhan2017complexity}
Salil Vadhan.
\newblock The complexity of differential privacy.
\newblock {\em Tutorials on the Foundations of Cryptography: Dedicated to Oded
  Goldreich}, pages 347--450, 2017.

\bibitem{vazirani2001approximation}
Vijay~V Vazirani.
\newblock {\em Approximation algorithms}, volume~1.
\newblock Springer, 2001.

\bibitem{vershynin2018high}
Roman Vershynin.
\newblock {\em High-dimensional probability: An introduction with applications
  in data science}, volume~47.
\newblock Cambridge university press, 2018.

\bibitem{vietriprivate}
Giuseppe Vietri, Cedric Archambeau, Sergul Aydore, William Brown, Michael
  Kearns, Aaron Roth, Ankit Siva, Shuai Tang, and Steven Wu.
\newblock Private synthetic data for multitask learning and marginal queries.
\newblock In {\em Advances in Neural Information Processing Systems}, 2022.

\bibitem{villani2009optimal}
C{\'e}dric Villani.
\newblock {\em Optimal transport: old and new}, volume 338.
\newblock Springer, 2009.

\bibitem{von2004distance}
Ulrike von Luxburg and Olivier Bousquet.
\newblock Distance-based classification with {L}ipschitz functions.
\newblock {\em J. Mach. Learn. Res.}, 5(Jun):669--695, 2004.

\bibitem{wang2016differentially}
Ziteng Wang, Chi Jin, Kai Fan, Jiaqi Zhang, Junliang Huang, Yiqiao Zhong, and
  Liwei Wang.
\newblock Differentially private data releasing for smooth queries.
\newblock {\em The Journal of Machine Learning Research}, 17(1):1779--1820,
  2016.

\bibitem{weed2019sharp}
Jonathan Weed and Francis Bach.
\newblock Sharp asymptotic and finite-sample rates of convergence of empirical
  measures in wasserstein distance.
\newblock {\em Bernoulli}, 25(4 A):2620--2648, 2019.

\end{thebibliography}

\appendix

\section{Additional proofs}\label{sec: appendix proof}

\subsection{Proof of Proposition~\ref{prop:Ln(f)}}
\begin{proof}
    We will apply the chaining argument (see, e.g.,     \cite[Chapter 8]{vershynin2018high}) to deduce a bound similar to Dudley's inequality.
    
    \noindent {\bf Step 1: (Finding nets)} 
    
    Define $\varepsilon_j = 2^{-j}$ for $j\in \mathbb{Z}$ and consider an $\varepsilon_j$-net $T_j$ of $\cF$ of size $\cN(\cF,\varepsilon_j,\|\cdot\|_{\infty})$. Then for any $f\in \cF$ and any level $j$, we can find the closest element in the net, denoted $\pi_j(f)$. In other words, there exists $\pi_j(f)$ s.t.
    \[\pi_j(f)\in T_j,\quad \|f-\pi_j(f)\|_\infty\leq \e_j.\]
   Let $m$ be a positive integer to be determined later, we have the telescope sum together with triangle inequality
    \begin{align*}
        \E\sup_{f\in \cF} \frac{1}{n}\left|\sum_{i=1}^{n}f(X_i)\lambda_i\right|\leq &\,\E\sup_{f\in \cF} \frac{1}{n}\left|\sum_{i=1}^{n}\left(f-\pi_m(f)\right)(X_i)\cdot\lambda_i\right|\\
        &\, + \sum_{j=j_0+1}^m \E\sup_{f\in \cF} \frac{1}{n}\left|\sum_{i=1}^{n}\left(\pi_{j}(f)-\pi_{j-1}(f)\right)(X_i)\cdot\lambda_i\right|.
    \end{align*}
    Note that when $j=j_0$ is small enough, $\Omega$ can be covered by $\pi_{j_0}(f) \equiv 0$.
    
    \noindent {\bf Step 2: (Bounding the telescoping sum)} 
    
    For a fixed $j_0 < j\leq m$, we consider the quantity
    \[\E\sup_{f\in \cF} \frac{1}{n}\left|\sum_{i=1}^{n}\left(\pi_{j}(f)-\pi_{j-1}(f)\right)(X_i)\cdot\lambda_i\right|.\]
    For simplicity we will denote $a_i=a_i(f)$ as the coefficient $\frac{1}{n}\left(\pi_{j}(f)-\pi_{j-1}(f)\right)(X_i)$. Then we have 
    \[|a_i|\leq \frac{1}{n}\left\|f-\pi_{j-1}(f)\right\|_\infty+\frac{1}{n}\left\|\pi_{j}(f)-f\right\|_\infty\leq \frac{1}{n}(\e_j+\e_{j-1})\leq \frac{3\e_j}{n}.\]
    
Since $\{\lambda_i\}_{i\in [n]}$ are independent subexponential random variables,   we can apply Bernstein's inequality to the sum $\sum_i a_i \lambda_i$. Let $K=3\e_j $, we have 
    \begin{align*}
        \Pr{\left|\sum_{i=1}^n a_i\lambda_i\right|>t}&\leq 2\exp\left[-c \min\left(\frac{t^2}{ \|a\|_2^2} , \frac{t}{\|a\|_\infty}\right)\right]\\
        &\leq 2\exp\left[-c \min\left(\frac{t^2}{K^2/n} , \frac{t}{K/n}\right)\right]\\
        & = 2\exp\left[-cn \min\left(\frac{t^2}{K^2} , \frac{t}{K}\right)\right],
    \end{align*}
    Then we can use the union bound to control the supreme. Define $N = |T_j|\cdot|T_{j-1}|\leq |T_j|^2$,
    \begin{align*}
        \Pr{\sup_{f\in\cF}\left|\sum_{i=1}^n a_i\lambda_i\right|>t}&\leq 2N\exp\left[-cn \min\left(\frac{t^2}{K^2} , \frac{t}{K}\right)\right]\wedge 1\\
        & = 2\exp\left[\log N-cn \min\left(\frac{t^2}{K^2} , \frac{t}{K}\right)\right]\wedge 1\\
        & \leq  2\exp\left(\log N-cn \frac{t^2}{K^2}\right)\wedge 1 \\
            & \quad  +  2\exp\left(\log N-cn \frac{t}{K}\right)\wedge 1
    \end{align*} 
    and hence
    \begin{align*}
        \E \sup_{f\in\cF} \left|\sum_{i=1}^n a_i\lambda_i\right|=& \int_{0}^\infty 2\exp\left(\log N-cn \frac{t^2}{K^2}\right)\wedge 1 \dd t \\
            & \quad+ \int_{0}^\infty 2\exp\left(\log N-cn \frac{t}{K}\right)\wedge 1 \dd t\\
        :=& \,I_2+I_1.
    \end{align*}
    We will compute them separately. 
    \begin{align*}
        I_1 &= \int_{0}^\infty 2\exp\left(\log N-cn \frac{t}{K}\right)\wedge 1 \dd t\\
        & = \frac{K \log N}{cn} + \int_{{K \log N}/{cn}}^\infty 2\exp\left(\log N-cn \frac{t}{K}\right)\\
        & = \frac{K \log N}{cn} + \int_{0}^\infty 2\exp\left(-cn \frac{t}{K}\right)\\
        &\leq CK\frac{\log N}{n}
    \end{align*}
    \begin{align*}
        I_2 &= \int_{0}^\infty 2\exp\left(\log N-cn \frac{t^2}{K^2}\right)\wedge 1 \dd t\\
        & = \sqrt{\frac{K^2 \log N}{cn}} + \int_{\sqrt{{K^2 \log N}/{cn}}}^\infty 2\exp\left(\log N-cn \frac{t^2}{K^2}\right)\\
        & = \sqrt{\frac{K^2 \log N}{cn}} + \int_{0}^\infty 2\exp\left(-cn \frac{t^2}{K^2}-2\sqrt{cn\log N}\frac{t}{K}\right)\\
        &\leq \sqrt{\frac{K^2 \log N}{cn}} + \frac{K}{\sqrt{cn\log N}}\\
        &\leq CK\sqrt{\frac{\log N}{n}}.
    \end{align*}
    Therefore we concluded that for a fixed level $j$, 
    \[\E \sup_{f\in\cF} \left|\sum_{i=1}^n a_i\lambda_i\right|\leq CK\left({\frac{\log N}{n}}+\sqrt{\frac{\log N}{n}}\right) \lesssim \e_j\left({\frac{\log N}{n}}+\sqrt{\frac{\log N}{n}}\right)\]
    
    \noindent {\bf Step 3: (Bounding the last entry)}
    
    For the last entry in the telescoping sum, similarly, we denote $a_i:=\frac{1}{n}\left(f-\pi_{m}(f)\right)(X_i)$ and we have $|a_i|\leq \e_m/n$. Then
    \[\sup_{f\in\cF} \left|\sum_{i=1}^n a_i\lambda_i\right| \leq \frac{\e_m}{n}\sum_{i=1}^n |\lambda_i|,\]
    and the expectation satisfies
    \[\E\sup_{f\in\cF} \left|\sum_{i=1}^n a_i\lambda_i\right| \leq \frac{\e_m}{n}\sum_{i=1}^n \E|\lambda_i| \lesssim \e_m .\]
    
    \noindent {\bf Step 4: (Combining the bound and choosing $m$)}
    Combining the two integrals together, we deduce that for any  $X_1,\dots,X_n\in \Omega$,
    \begin{align*}
        \E\sup_{f\in \cF} \frac{1}{n}\left|\sum_{i=1}^{n}f(X_i)\lambda_i\right|\leq C \left(\e_m + \sum_{j=j_0+1}^m\right. &\e_j\Bigg({\frac{\log \cN(\cF,\varepsilon_j,\|\cdot\|_{\infty})}{n}} \\
        &\left.+\sqrt{\frac{\log \cN(\cF,\varepsilon_j,\|\cdot\|_{\infty})}{n}}\Bigg)\right).
    \end{align*}
    Then for any $\alpha>0$, we can always choose $m$ such that $2\alpha\leq \e_m < 4\alpha$ and bound the sum above with integral
    \begin{align}
        \label{eq: entropy integral}
        \E\sup_{f\in \cF} \frac{1}{n}\left|\sum_{i=1}^{n}f(X_i)\lambda_i\right| \leq  C\Bigg(2\alpha &+ \frac{1}{\sqrt{n}}\int_{\alpha}^\infty \sqrt{\log \cN(\cF,u,\|\cdot\|_{\infty})}\dd u \notag\\
        &+\frac{1}{{n}}\int_{\alpha}^\infty {\log \cN(\cF,u,\|\cdot\|_{\infty})}\dd u\Bigg).
    \end{align}
	Taking infimum over $\alpha$ completes the proof of the first inequality.

Now assume $\mathcal F$ is the set of all functions $f$ with $\|f\|_{\Lip}\leq 1$.
  From     \cite[Lemma 4.2]{gottlieb2016adaptive}, we can bound the covering number of $\mathcal F$ by the covering number of $\Omega$ as follows:
    	\[\log \cN(\cF,u,\|\cdot\|_{\infty}) \leq \log ({8}/{u})\, \cN(\Omega, {u}/{2},\rho).\]
    As a result, for any $\alpha>0$,
    \[L_{n}(\mathcal F)\leq C \left(2\alpha + \frac{1}{\sqrt{n}}\int_{\alpha}^{\infty} \sqrt{\log ({8}/{u})\, \cN(\Omega, {u}/{2},\rho)}\dd u+\frac{1}{{n}}\int_{\alpha}^{\infty} \log ({8}/{u})\, \cN(\Omega, {u}/{2},\rho)\dd u\right).\]	
        This completes the proof.
\end{proof}

\subsection{Proof of Corollary~\ref{cor: hypercube}}
\begin{proof}
     For $\Omega=[0,1]^d$ with $l_{\infty}$-norm, we have $\diam(\Omega) = 1$ and the covering number 
    \[\cN([0,1]^d, u, \|\cdot\|_\infty) \leq u^{-d}.\]
    Then, as the domain $\Omega = [0,1]^d$ is connected and centered, we can apply the bound for the covering number of $\mathcal F$ from     \cite[Theorem 17]{von2004distance}:
    \[\cN(\cF,u,\|\cdot\|_{\infty})\leq \left(2\left\lceil{2}/u \right\rceil+1\right) 2^{\cN\left([0,1]^d, u/2, \|\cdot\|_\infty\right)},\]
    \[\Longrightarrow \log \cN(\cF,u,\|\cdot\|_{\infty}) \lesssim \cN(\Omega, {u}/{2},\|\cdot\|_\infty)\lesssim (u/2)^{-d}.\]
  Applying the inequality above to \eqref{eq: entropy integral}, we get
    \begin{equation}
    \label{eq: Dudley integral}
        L_{n}\leq C \left(2\alpha + \frac{1}{\sqrt{n}}\int_{\alpha}^{\infty} (u/2)^{-d/2}\dd u+\frac{1}{{n}}\int_{\alpha}^{\infty} (u/2)^{-d}\dd u\right).
    \end{equation}

    Compute the integral for the case $d=2$ and $d\geq 3$,
        \[L_{n}(f)\leq \left\{\begin{aligned}
        &C \Bigg(2\alpha + \frac{2}{\sqrt{n}}\log\frac2{\alpha} + \frac{2}{n}\left(\frac{\alpha}{2}\right)^{-1}\Bigg) &\quad \textrm{ if } d=2.\\
        &C \Bigg(2\alpha + \frac{2}{\sqrt{n}}\cdot\frac{1}{\frac{d}{2}-1}\left(\frac{\alpha}{2}\right)^{1-\frac{d}{2}} + \frac{2}{n}\cdot\frac{1}{d-1}\left(\frac{\alpha}{2}\right)^{1-d}\Bigg) \quad &\textrm{ if } d\geq 3.
    \end{aligned}\right.\]
    Choosing $\alpha = 2n^{-1/d}$ finishes the cases for $d\geq 2$.

    When $d=1$, the Dudley integral in \eqref{eq: Dudley integral} is divergent. However, note that $\diam(\cF)\leq 2$ and hence $\log\cN(\cF,u,\|\cdot\|_{\infty})=0$ for $u>1$. From \eqref{eq: entropy integral}, we have 
    \begin{align*}
        L_n(\cF) &\leq C \left(2\alpha + \frac{1}{\sqrt{n}}\int_{\alpha}^{1} (u/2)^{-1/2}\dd u+\frac{1}{{n}}\int_{\alpha}^{1} (u/2)^{-1}\dd u\right)\\
        &\leq C \left(2\alpha + \frac{2(\sqrt{2}-\sqrt{\alpha})}{\sqrt{n}}+\frac{2}{n}\log\frac{1}{\alpha}\right).
    \end{align*}
    The optimal choice of $\alpha$ is  $\alpha\sim n^{-1/2}$, which gives us the result for $d=1$.
\end{proof}

\subsection{Proof of Proposition~\ref{prop: direct_priv}}
\begin{proof}
It suffices to prove that the steps from $\cX$ to the sign measure $\nu$ in Algorithm~\ref{alg: direct} is $\e$-differentially private since the remaining steps are only based on $\nu$. Notice that both $\mu_\cY, \nu$ are supported on $Y_1,\dots, Y_m$,  we can identify the two discrete measures  as $m$ dimensional vectors in the standard simplex, denoted $\overline{\mu_\cY}, \overline{\nu}$, respectively. Consider two data sets $\cX_1$ and $\cX_2$ differ in one point. Suppose we deduced  ${\mu_{\cY_1}}, {\mu_{\cY_2}}$ and $\nu_1,\nu_2$ through the first four steps of Algorithm~\ref{alg: direct} from $\cX_1, \cX_2$, respectively. We know two vectors $\overline{\mu_{\cY_1}}, \overline{\mu_{\cY_2}}$ are different  at  one coordinate, where the difference is bounded by $1/n$. 

Then
    \begin{align*}
        \frac{\Pr{\nu_1=\eta}}{\Pr{\nu_2=\eta}} & = \prod_{i=1}^{m}\frac{\Pr{\lambda_i=n(\eta-\overline{\mu_{\cY_1}})_i}}{\Pr{\lambda_i=n(\eta-\overline{\mu_{\cY_2}})_i}} = \prod_{i=1}^{m}\frac{\exp(-\e n|(\eta-\overline{\mu_{\cY_1}})_i|)}{\exp(-\e n|(\eta-\overline{\mu_{\cY_2}})_i|)}\\
        &\leq \exp\left(\e n \|{\mu_{\cY_2}} - {\mu_{\cY_1}}\|_1\right)  \leq e^\e.
    \end{align*}
By writing $\Pr{\nu_i\in S}=\sum_{\eta\in S}\Pr{\nu_i=\eta}$ for $i=1,2$, the inequality above implies Algorithm~\ref{alg: direct} is $\e$-differentially private.
\end{proof}

\subsection{Proof of Proposition~\ref{prop:LP}}
\begin{proof}
	For two signed  measures $\tau,\nu$ supported on $\cY$, the $\dbl$-distance between $\tau$ and $\nu$ is 
	\[\dbl(\tau,\nu) = \sup_{\|f\|_{\mathrm{Lip}\leq 1}} \left|\sum_{i=1}^m f(y_i) \left(\tau(\{y_i\}) - \nu(\{y_i\})\right)\right|.\]
	For simplicity, we denote  $f_i = f(y_i)$, $\nu_i = \nu(\{y_i\})$ and $\tau_i = \tau(\{y_i\})$. Then we note that for any $f$ with $\|f\|_{\mathrm{Lip}}\leq 1$, only $(f_i)_{i\in [m]}$ matters in the definition above. Therefore, suppose $\nu$ and $\tau$ are fixed,  computing the   $\dbl$-distance is equivalent to the following linear programming problem:
 
	\begin{align*}
		\max\quad &  \sum_{i=1}^m (\nu_i-\tau_i)f_i \\
		\mbox{s.t.}\quad & f_i - f_j\leq \|y_i - y_j\|_\infty, & {\forall i,j\leq m, i\not=j},\\
		& -f_i +  f_j\leq \|y_i - y_j\|_\infty, & \forall i,j\leq m, i\not=j,\\
		& -1 \leq f_i \leq 1, & \forall i\leq m.
	\end{align*}
	After a change of variable $f_i' = f_i + 1$, we can rewrite it as
	\begin{align*}
	 \max\quad & \sum_{i=1}^m (\nu_i-\tau_i)f_i' - \left(\nu(\Omega)-1\right) \\
	\mbox{s.t.}\quad & f_i' - f_j'\leq \|y_i - y_j\|_\infty, & \forall i,j\leq m, i\not=j,\\
	& -f_i' +  f_j'\leq \|y_i - y_j\|_\infty, & \forall i,j\leq m, i\not=j,\\
	& 0\leq f_i' \leq 2, & \forall i\leq m.
	\end{align*}
	
 Next, we can consider the dual problem of the linear programming problem above. The duality theory in linear programming   \cite[Chapter 12]{vazirani2001approximation} showed that the original problem and the dual problem have the same optimal solution. Let $u_{ij}, u_{ij}'\geq 0$ be the dual variable for the linear constraints about $f_i'-f_j'$ and $-f_i'+f_j'$, and let $v_i\geq 0$ be the dual variable for the equation $f_i'\leq 2$. As the linear programming above is in the standard form, by the duality theory, it is equivalent to
	\begin{align*}
	\min \quad & \sum_{i\neq j}  \|y_i - y_j\|_\infty(u_{ij}+u_{ij}') + 2\, v_i - \left(\nu(\Omega)-1\right)\\
	\mbox{s.t.}\quad & \sum_{j\neq i} (u_{ij} - u_{ij}') + v_i\geq \nu_i-\tau_i, & \forall {i\leq m},\\
	& u_{ij},u_{ij}',v_i \geq 0  & \forall i,j\leq m, i\not=j.
	\end{align*}
	
To find the minimizer $\tau$ for a given $\nu$, we regard $\tau_i$ as variables and add the constraints of $\tau$ being a probability measure. Also, we can eliminate the constant $\nu(\Omega)-1$ in the target function. So we get the linear programming problem:
	\begin{align*}
		\min \quad & \sum_{i\neq j}  \|y_i - y_j\|_\infty(u_{ij}+u_{ij}') + 2\, v_i\\
		\mbox{s.t.}\quad & \sum_{j\neq i} (u_{ij} - u_{ij}') + v_i +\tau_i \geq \nu_i, & \forall {i\leq m},\\
		& \sum_{i=1}^m \tau_i = 1, &\\ 
		& u_{ij},u_{ij}',v_i, \tau_i \geq 0  & \forall i,j\leq m, i\not=j.
	\end{align*}
There are $2m^2$ variables in total and $m+1$ linear constraints, and the minimizer $(\tau_i)_{i=1}^m$ is what we want.  
\end{proof}

\subsection{Proof of Theorem~\ref{thm: accuracy_direct}}

\begin{proof}
    We transformed the original data measure $\mu_\cX$ with three steps: 
    $\mu_\cX\longrightarrow \mu_\cY\longrightarrow \nu \longrightarrow \hat{\nu}.$ 
    
    \noindent{\bf Step 1:} For the first step in the algorithm, we have $W_1(\mu_\cX,\mu_\cY) \leq \max_i \diam(\Omega_i)$. This follows from the definition of $1$-Wasserstein distance.
    
    \noindent{\bf Step 2:} In this step, $\nu$ is no longer a probability measure, and we  consider $\dbl(\mu_\cY,\nu)$ instead: 
    \begin{align}
        \mathbb E\dbl(\mu_\cY,\nu) &= \mathbb E\sup_{\|f\|_{\mathrm{Lip}}\leq 1}\left|\int f\dd\mu_\cY - \int f\dd\nu\right| \nonumber\\
        &= \mathbb E\sup_{\|f\|_{\mathrm{Lip}}\leq 1}\left|\sum_{i=1}^m f(y_i)\left(\frac{n_i}{n} + \frac{\lambda_i}{n} - \frac{n_i}{n}\right)\right| =  \frac{m}{\e n} \widetilde L_m(\mathcal F).\label{eq:dBL_Laplace}
    \end{align}
    
    \noindent {\bf Step 3:} For the last step, we have $\dbl(\nu,\hat{\nu}) \leq \dbl(\mu_\cY,\nu)$ because  $\hat{\nu}$ is the closest probability measure to $\nu$ from Proposition~\ref{prop:LP}. As a result, we have
    \begin{align*}
    W_1(\mu,\hat\nu)= \dbl(\mu,\hat \nu)& \leq \dbl(\mu_\cX,\mu_\cY) +\dbl(\mu_\cY,\nu) + \dbl(\nu,\hat{\nu})\\
    &\leq W_1(\mu_\cX,\mu_\cY) + 2\dbl(\mu_\cY,\nu)\leq \max_{i} \diam(\Omega_i) +2\dbl(\mu_\cY,\nu).
    \end{align*}
    After taking the expectation, we can apply \eqref{eq:dBL_Laplace} to get the desired inequality.
\end{proof}

\subsection{Proof of Corollary \ref{cor: accuracy_direct_hypercube}}
\begin{proof}
Using Theorem~\ref{thm: accuracy_direct}, we have 
\[\E W_1(\mu_\XX,\hat\nu)\leq \max_i \diam(\Omega_i) + \frac{2 m}{\e n} \widetilde L_m(\cF).\]
By assumption we have $\max_i \diam(\Omega_i)\asymp m^{-1/d} \asymp (\e n)^{-1/d}$. And by \ref{cor: hypercube} we have the bound for the Laplacian complexity
    \[\widetilde L_{m}(\mathcal F)  \leq \left\{\begin{aligned}
    & C (\e n)^{-1/2}\quad &\textrm{ if } d=1,\\
     &{C\log n} \cdot (\e n)^{-1/2}\quad &\textrm{ if } d=2,\\
    &{C} {(\e n)^{-{1}/{d}}}\quad &\textrm{ if } d\geq 3.
    \end{aligned}\right.\]
When $d\geq 3$, the two terms are comparable. And when $d=1,2$, the Laplacian complexity dominates the error. Combining the two inequalities gives the result.
\end{proof}

\subsection{Proof of Theorem~\ref{thm: privacy}}

 {Theorem \ref{thm: privacy} can be obtained by applying the parallel composition lemma \cite{dwork2014algorithmic}.} Here we present a self-contained proof by considering an inhomogeneous version of the classical Laplacian mechanism \cite{dwork2014algorithmic}.

\begin{lemma}[Inhomogeneous Laplace mechanism]		\label{lem: Laplace mechanism}
  Let $F: \Omega^n \to \R^k$ be any map, 
  $s = (s_i)_{i=1}^k \in \R_+^k$ be a fixed vector,
  and $\l = (\l_i)_{i=1}^k$ be a random vector 
  with independent coordinates $\l_i \sim \Lap_\Z(s_i)$.
  Then the map $x \mapsto F(x)+\l$ is $\e$-differentially private, where
  $$
  \e = \sup_{x,\tilde{x}} \norm{F(x)-F(\tilde{x})}_{\ell^1(s)}.
  $$
  Here the supremum is over all pairs of input vectors in $\Omega^n$ 
  that differ in one coordinate, and 
  $\norm{z}_{\ell^1(s)} = \sum_{i=1}^k \abs{z_i}/{s_i}$.
\end{lemma}
\begin{proof}[Proof of Lemma~\ref{lem: Laplace mechanism}]
Suppose $x,\tilde{x}\in \Omega^n$ differs in exactly one coordinate. Consider the density functions of the inputs having the same output $y = F(x)+\lambda = F(\tilde{x})+\tilde{\lambda}\in \Z^k$. We have
\begin{align*}
    \frac{\Pr{F(x)+\lambda = y}}{\Pr{F(\tilde{x})+\tilde{\lambda} = y}} &= \frac{\Pr{\lambda = y-F(x)}}{\Pr{{\tilde{\lambda}}={y-F(\tilde{x})}}} \\ 
    & = \frac{\prod_{i=1}^k \exp\left(-\frac{|(y-F(x))_i|}{s_i}\right) }{\prod_{i=1}^k \exp\left(-\frac{|(y-F(\tilde{x}))_i|}{s_i}\right)} \\
    & = \exp\left(-\sum_{i=1}^k \frac{1}{s_i}\left(|(y-F(x))_i| - |(y-F(\tilde{x}))_i|\right)\right) \\
    & \leq \exp\left(\|F(x)-F(\tilde{x})\|_{\ell^1(s)}\right)\\
    &\leq e^\e
\end{align*}
Therefore, we know $x\mapsto F(x)+\lambda$ is $\e$-differentially private.
\end{proof}

\begin{proof}[Proof of Theorem~\ref{thm: privacy}]
Consider the map $F(\XX)=(n_\theta)$ that transforms the input data into the vector of counts. 
Suppose a pair of input data $\XX$ and $\tilde{\XX}$ differ in one point $x_i$.
Consider the corresponding vectors of counts $(n_\theta)$ and $(\tilde{n}_\theta)$. 
For each level $j=0,\ldots,r$, the vectors of counts differ for a single $\theta \in \{0,1\}^j$, 
namely for the $\theta$ that corresponds to the region $\Omega_\theta$ containing $x_i$.
Moreover, whenever such a difference occurs, we have $\abs{n_\theta-\tilde{n}_\theta}=1$. 
Thus, extending the vector $(\s_j)_{j=0}^r$ to $(\s_\theta)_{\theta \in \{0,1\}^{\le r}}$
trivially (by converting $\s_j$ to $\s_\theta$ for all $\theta \in \{0,1\}^j$), we have
$$
\norm{F(\XX)-F(\tilde{\XX})}_{\ell^1(\s)}
=\sum_{j=0}^r \frac{1}{\s_j} \sum_{\theta \in \{0,1\}^j} \abs{n_\theta-\tilde{n}_\theta}
=\sum_{j=0}^r \frac{1}{\s_j}
= \e.
$$
Applying Lemma~\ref{lem: Laplace mechanism}, we conclude that the map 
$\XX \mapsto (n_\theta+\l_\theta)$ is $\e$-differentially private.
\end{proof}

\subsection{Proof of Theorem~\ref{thm: optimized accuracy}}
\begin{proof}
We will use the Lagrange multipliers procedure to find the optimal choices of $\sigma_j$. Given the maximal layer $r$, recall Theorem~\ref{thm: privacy}, we should use our privacy budget as
\[\e = \sum_{j=0}^r \frac{1}{\sigma_j}.\]
Therefore, we aim to minimize the accuracy bound with the specified privacy budget, namely
\begin{align*}
    &\min \E W_1(\mu_\cX,\mu_\cY)\quad  
    \;\textrm{s.t. }\e = \sum_{j=0}^r \frac{1}{\sigma_j}.
\end{align*}
Recall the result in Theorem~\ref{thm: accuracy}. Here $\e,n$ are given and $\delta$ is fixed as long as we determine the maximal level $r$. So the minimization problem is
\begin{align*}
    &\min \sum_{j=0}^r \sigma_j\Delta_{j-1} \quad 
    \;\textrm{s.t. }\e = \sum_{j=0}^r \frac{1}{\sigma_j}.
\end{align*}
Consider the Lagrangian function
\[f(\sigma_0,\dots,\sigma_r; t) := \sum_{j=0}^r \sigma_j\Delta_{j-1} - t\left(\sum_{j=0}^r \frac{1}{\sigma_j} - \e\right)\]
and the corresponding equation 
\[\frac{\partial f}{\partial \sigma_0} = \dots = \frac{\partial f}{\partial \sigma_r} = \frac{\partial f}{\partial t}=0.\]

One can easily check that the equations above have a unique solution 
\begin{equation}	\label{eq: optimal noise}
\s_j = \frac{S}{\e \sqrt{\Delta_{j-1}}}
\quad \text{where} \quad
S = \sum_{j=0}^r \sqrt{\Delta_{i-1}}.
\end{equation}
and it is indeed a minimal point for $f(\sigma_0,\dots,\sigma_r; t)$.

As a result, if we fix $\e$ and want Algorithm~\ref{alg: synthetic data} to be $\e$-differentially private, we should choose the noise magnitudes as \eqref{eq: optimal noise}. Substituting these noise magnitudes into the accuracy Theorem~\ref{thm: accuracy}, 
we see that the accuracy gets bounded by $\frac{\sqrt{2}}{\e n} S^2  + \d$.
\end{proof}

\subsection{Proof of Corollary~\ref{cor: PMM_hypercube}}
\begin{proof}
Let $\Omega=[0,1]$ with the $\ell^{\infty}$ metric. The natural hierarchical binary decomposition of $[0,1]$ (cut through the middle) makes subintervals of length $\diam(\Omega_\theta) = 2^{-j}$ 
for $\theta \in \{0,1\}^j$, so $\Delta_j = 1$ for all $j$, and the resolution is $\d = 2^{-r}$.
Theorem~\ref{thm: optimized accuracy} makes  $\e$-differential private synthetic data
with accuracy
$$
\E W_1 \left( \mu_\XX, \mu_\YY \right) 
\le \frac{\sqrt{2}(r+1)^2}{\e n}+2^{-r}.
$$
A nearly optimal choice for $r$ is $r = \log_2(\e n)-1$, which yields
$$
\E W_1 \left( \mu_\XX, \mu_\YY \right) 
\le \frac{(2+\sqrt{2}) \log_2^2(\e n)}{\e n}.
$$
The optimal noise magnitudes, per \eqref{eq: optimal noise}, are $\s_j = \log_2^2(\e n)/\e$. In other words, the  noise \textit{does not decay} with the level.


Let $\Omega=[0,1]^d$ for $d>1$. The natural hierarchical binary decomposition of $[0,1]^d$
(cut through the middle along a coordinate hyperplane)
makes subintervals of length $\diam(\Omega_\theta) \asymp 2^{-j/d}$ 
for $\theta \in \{0,1\}^j$, so $\Delta_j = 2^j \cdot 2^{-j/d} = 2^{(1-1/d)j}$ for all $j$, 
and the resolution is $\d = 2^{-r/d}$. Thus, 
$$
S = \sum_{j=0}^r \sqrt{\Delta_{j-1}}
\sim 2^{\frac12(1-\frac1d)r}.
$$
Theorem~\ref{thm: optimized accuracy} makes a $\e$-differential private synthetic data
with accuracy
$$
\E W_1 \left( \mu_\XX, \mu_\YY \right) 
\lesssim \frac{2^{(1-\frac1d)r}}{\e n} + 2^{-r/d}.
$$
A nearly optimal choice for the depth of the partition is  $r = \log_2(\e n)$,
which yields
$$
\E W_1 \left( \mu_\XX, \mu_\YY \right) 
\lesssim (\e n)^{-1/d}.
$$
 The optimal noise magnitudes, per \eqref{eq: optimal noise}, are 
$$
\s_j \sim \e^{-1} 2^{\frac12(1-\frac1d)(r-j)}.
$$
Thus, the {\em noise decays} with the level $j$, becoming $O(1)$ per region for the smallest regions.
\end{proof}

\subsection{Proof of Lemma~\ref{lem: flux incomparability}}
\begin{proof}
If $a,b$ are comparable, both values are zero. If $a,b$ is not comparable, we can assume $a_1>b_1$, $a_2<b_2$ without loss of generality. The set of points that are comparable to $b$ is 
\[\{(x_1,x_2)\in \Z_+^2\mid x_1\leq b_1,x_2\leq b_2\}\cup \{(x_1,x_2)\in \Z_+^2\mid x_1\geq b_1,x_2\geq b_2\}.\]
Note that the distance from $a$ to the first set is $\abs{a_1-b_1}$ and the distance from $a$ to the second set is $\abs{a_2-b_2}$. Then $\flux(a,b)$ is the smaller one of the two distances, which is also the distance from $a$ to the union set.
\end{proof}

\subsection{Proof of Lemma~\ref{lem: flux transfer}}
\begin{proof}
{\em Case 1: $a=(a_1,a_2)$ and $b=(b_1,b_2)$ are comparable.}
If $a \preceq b$, remove $b_1-a_1$ balls from bin 1 and $b_2-a_2$ balls from bin 2 to achieve the result. If $b \preceq a$, adding $a_1-b_1$ balls to bin 1 and $a_2-b_2$ balls to bin 2 to achieve the result.

{\em Case 2: $a=(a_1,a_2)$ and $b=(b_1,b_2)$ are incomparable.}
Without loss of generality, we can assume that $a_1-b_1 \ge 0$, $a_2-b_2 \le 0$.

Assume first that $a_1-b_1 \ge b_2-a_2$. Then $\flux(a,b)=b_2-a_2 \coloneqq M$.
Then $\Delta \eqqcolon (a_1+a_2)-(b_1+b_2)>0$. Removing $\Delta$ balls from bin 1 and transferring $M$ balls from bin 1 to bin 2 achieves the result. Note that there are enough balls in bin $1$ to transfer, since $M+\Delta = a_1-b_1 \in [0,a_1]$.

Now assume that $a_1-b_1 \le b_2-a_2$. Then $\flux(a,b)=a_1-b_1 \coloneqq M$.
Then $\Delta \eqqcolon (b_1+b_2)-(a_1+a_2)>0$. Adding $\Delta$ balls to bin 2 and transferring $M$ balls from bin 1 to bin 2 achieves the result. 
\end{proof}

\subsection{Proof of Lemma~\ref{lem: tranform data}}
\begin{proof}
First, we make the total number of points in $\Omega$ correct by adding $m-n$ points to $\Omega$ (or removing, if that number is negative). 

Apply Lemma~\ref{lem: flux transfer} for the two parts of $\Omega$: bin $\Omega_0$ that contains $n_0$ points and bin $\Omega_1$ that contains $n_1$ points. Since $\Omega$ already contains the correct total number of points $m$, we can make the two bins contain the correct number of points, i.e. $m_0$ and $m_1$ respectively, by transferring $\flux \left( (n_0,n_1), \, (m_0,m_1) \right)$ points from one bin to the other.

Apply Lemma~\ref{lem: flux transfer} for the two parts of $\Omega_0$: bin $\Omega_{00}$ that contains $n_{00}$ points and bin $\Omega_{01}$ that contains $n_{01}$ points. 
Since $\Omega_0$ already contains the correct number of points $m_0$, we can make the two bins contain the correct number of points, i.e. $m_{00}$ and $m_{01}$ respectively, by transferring $\flux \left( (n_{00},n_{01}), \, (m_{00},m_{01}) \right)$ points from one bin to the other.

Similarly, since $\Omega_1$ already has the correct number of points $m_1$,  we can make  $\Omega_{10}$ and $\Omega_{11}$ contain the correct number of points $m_{10}$ and $m_{11}$ by transferring $\flux \left( (n_{00},n_{01}), \, (m_{00},m_{01}) \right)$ points from one bin to the other.

Continuing this way, we can complete the proof. 
Note that the steps of the iteration procedure we described are interlocked. 
Each next step determines which subregion the transferred points are selected from, and which subregion they are moved to in the previous step. For example, the original step calls to add (or remove) $m-n$ points to or from $\Omega$, but does not specify how these points are distributed between the two parts $\Omega_0$ and $\Omega_1$. The application of Lemma~\ref{lem: flux transfer} at the next step determines this.
\end{proof}
\subsection{Proof of Lemma \ref{lem: flux and laplacian}}
\begin{proof}
We will derive this result from Lemma~\ref{lem: flux incomparability}.
First, let us compute the distance from $a = (n_{\theta0},n_{\theta1})$ to 
$b' = (n'_{\theta0},n'_{\theta1}) = \left( (n_{\theta0} + \l_{\theta0})_+, \, (n_{\theta1} + \l_{\theta1})_+ \right)$. Since the map $x \mapsto x_+$ is $1$-Lipschitz, we have
$$
\norm{a-b'}_\infty \le \max \left( \abs{\l_{\theta0}}, \abs{\l_{\theta1}} \right).
$$
Furthermore, recall that by Algorithm~\ref{alg: consistency}, $b'$ is comparable to 
$b = (m_{\theta0},m_{\theta1})$. 
An application of Lemma~\ref{lem: flux incomparability} completes the proof.
\end{proof}

\subsection{Proof of Lemma \ref{lem: adding points}}
\begin{proof}
Finding the 1-Wasserstein distance in the discrete case is equivalent to solving the optimal transformation problem. In fact, we can obtain $\mu_U$ from $\mu_V$ by moving $\abs{V \setminus U}$ atoms of $\mu_V$, each having mass $1/\abs{V}$, and distributing their mass uniformly over $U$. The distance for each movement is bounded by $\diam(\Omega)$. Therefore the $1$-Wasserstein distance between $\mu_{U}$ and $\nu_{V}$ is bounded by $\frac{|V\setminus U|}{|V|}\diam(\Omega)$.  
\end{proof}

\section{Discrete Laplacian distribution}\label{sec: appendix discrete}

Recall that the classical Laplacian distribution $\Lap_\R(\s)$ is a continuous distribution with density 
$$
f(x) = \frac{1}{2\s} \exp \left( -\abs{x}/\s \right),
\quad x \in \R.
$$
A random variable $X \sim \Lap_\R(\s)$ has zero mean and 
$$
\Var(Z) = 2\s^2.
$$
To deal with counts, it is more convenient to use the {\em discrete} 
Laplacian distribution $\Lap_\Z(\s)$, see     \cite{IK}, which has probability mass function
$$
f(z) = \frac{1-p_\s}{1+p_\s} \exp \left( -\abs{z}/\s \right),
\quad z \in \Z
$$
where $p_\s = \exp(-1/\s)$.
A random variable $Z \sim \Lap_\Z(\s)$ has zero mean and 
$$
\Var(Z) = \frac{2p_\s}{(1-p_\s)^2}.
$$
Thus, one can verify that discrete Laplacian has a smaller variance than its continuous counterpart:
\begin{equation}	\label{eq: Lap variance}
\Var(Z) < 2\s^2,
\end{equation}
but the gap vanishes for large $\s$:
$$
\Var(Z) \to 2\s^2 
\quad \text{as } \s \to \infty.
$$
\end{document}